\documentclass[runningheads]{llncs}
\usepackage{graphicx}
\usepackage{epsfig,epsf}
\usepackage{epstopdf}
\usepackage{graphics}
\usepackage{array}
\usepackage{amsmath}
\usepackage{amssymb}
\usepackage{comment}
\usepackage{algorithm}
\usepackage{algorithmic}

\usepackage{fullpage}
\date{}

\title{On $P_5$-free Chordal Bipartite graphs}
\author{S.Aadhavan, P.Renjith, N.Sadagopan} 
\institute{Indian Institute of Information Technology Design and Manufacturing, Kancheepuram, Chennai. \\
\email{aadhav1395@gmail.com, \{coe14d002, sadagopan\}@iiitdm.ac.in}}
\begin{document}
\maketitle
\begin{abstract}
A bipartite graph is chordal bipartite if every cycle of length at least 6 has a chord in it.  In this paper, we investigate the structure of $P_5$-free chordal bipartite graphs and show that these graphs have a {\em Nested Neighborhood Ordering}, a special ordering among its vertices.  Further, using this ordering, we present polynomial-time algorithms for classical problems such as Hamiltonian cycle (path) and longest path.  Two variants of the Hamiltonian path problem include the Steiner path and the minimum leaf spanning tree problems, and we obtain polynomial-time algorithms for these problems as well restricted to $P_5$-free chordal bipartite graphs. \newline\newline
\end{abstract}
\section{Introduction}
The graphs with forbidden subgraphs possess a nice structural characterization.  The structural characterization of these graphs has attracted  researchers from both mathematics and computing. The popular graphs are chordal graphs\cite{dirac1961rigid}, which forbid induced cycles of length at least four, and chordal bipartite graphs\cite{fulkerson1965incidence}, which are bipartite graphs that forbid induced cycles of length at least six.  Further, these graph classes have a nice structural characterization with respect to minimal vertex separators and a special ordering, namely, perfect elimination ordering\cite{golumbie1980algorithmic} among its vertices (edges).  These graphs are largely studied in the literature to understand the computational complexity of classical optimization problems such as vertex cover, dominating set, coloring, etc., as these problems are known to be NP-complete in general graphs.  Thus, these graphs help to identify the gap between NP-completeness and polynomial-time solvable instances of a combinatorial problem.   \\
The Hamiltonian cycle (path) problem is a famous problem which asks for the presence of a cycle (path) that visits each node exactly once in a graph. Hamiltonian problems play a significant role in various research areas such as circuit design\cite{wang2012efficient}, operational research\cite{malakis1976hamiltonian}, biology\cite{dorninger1994hamiltonian}, etc.  On the complexity front, this problem is well-studied and it remains NP-complete on general graphs.  Interestingly, this problem is polynomial-time solvable on special graphs such as cographs and permutation graphs \cite{deogun1994polynomial}.  Surprisingly, this problem is NP-complete in chordal graphs\cite{bertossi1986hamiltonian}, bipartite graphs\cite{Krishnamoorthy:1975:NPB:990518.990521} and $P_5$ free graphs. \\
Haiko Muller\cite{muller1996hamiltonian}  has shown that the hamiltonian cycle problem is NP-complete in chordal bipartite graphs by a polynomial-time reduction from the satisfiablity problem .  The microscopic view of reduction instances reveals that the instances are $P_9$-free chordal bipartite graphs.  It is natural to study the complexity of hamiltonian cycle problem in  $P_8$-free chordal bipartite graphs and its subclasses.   Since $P_4$-free chordal bipartite graphs are complete bipartite graphs, the first non-trivial graph class in this line of research is $P_5$-free chordal bipartite graphs.   It is known from the literature that problems such as hamiltonian cycle, clique, clique cover, domination, etc.,  have polynomial-time algorithms in $P_4$-free graphs.  In recent times, the class $P_5$-free graphs have received a good attention, and  problems such as independent set\cite{Lokshantov:2014:ISP:2634074.2634117} and 3-colourbility\cite{randerath20013} have polynomial-time  algorithms restricted to $P_5$-free graphs. \\
{\bf Our work:} In this paper, we study the structure of $P_5$-free chordal bipartite graphs and present a new ordering referred to as {\em Nested Neighbourhood Ordering} among its vertices. We present a polynomial-time algorithm for hamiltonian cycle (path) problem using the Nested Neighbourhood Ordering.  Further, using this ordering, we present polynomial-time algorithms for longest path, minimum leaf spanning tree and Steiner path problems.
\subsection{Preliminaries}
All the graphs used in this paper are simple, connected and unweighted.  For a graph $G$, let $V(G)$ denote the vertex set and $E(G)$ denote the edge set.  The notation $uv$ represents an edge incident on the vertices $u$ and $v$.  The degree of a vertex $u$ is denoted as $d_G(u)=|N_G(u)|$, where $N_G(u)$  denotes the set of vertices that are adjacent to $u$.  If $G$ is disconnected, then $c(G)$ denotes the number of connected components in $G$ (each component being maximal).  A bipartite graph is chordal bipartite if every cycle of length six has a chord. A maximal biclique $K_{i,j}$ is a complete bipartite graph such that there is no strict supergraphs $k_{i+1,j}$ or $k_{i,j+1}$.   A maximum biclique is a maximal biclique with the property that $|i-j|$ is minimum. Note that $P_5$ is an induced path of length $5$ and $P_{uv}$ denotes a path that starts at $u$ and ends at $v$, and $|P_{uv}|$ denotes its length.   We use $P_{uv}$ and $P(u,v)$ interchangeably.
\section{Structural Results}
In this section, we shall present a structural characterization of $P_5$-free chordal bipartite graphs.  Also, we introduce {\em Nested Neighborhood Ordering} among its vertices.  We shall fix the following notation to present our results.  For a chordal bipartite $G$ with bipartition $(A,B)$, let $A=\{x_1,x_2,\ldots,x_m\}$ and $B=\{y_1,y_2,\ldots,y_n\}$. $E(G) \subseteq \{xy ~|~ x \in A,y \in B \}$.  We write $A=A_{1} \cup A_{2}$ and $B=B_{1} \cup B_{2}$,
$A_{1}=\{x_{1},x_{2},\ldots,x_{i}\}$, $A_{2}=\{x_{i+1},x_{i+2},\ldots,x_{m}\}$ and $B_{1}=\{y_1,y_2,\ldots,y_j\}$, $B_{2}=\{y_{j+1},y_{j+2},\ldots,y_{n}\}$, such that $(A_1,B_1)$ denotes the maximum biclique.  
\begin{lemma}
\label{lem1}
Let $G$ be a $P_5$-free chordal bipartite graph.  Then, $\forall x \in A_{2},\exists y_{k} \in B_{1}$ such that $y_{k} \not\in N_G(x)$ and $\forall y \in B_{2},\exists x_{k} \in A_{1}$ such that $x_{k} \not\in N_G(y)$.
\end{lemma}
\begin{proof}
Suppose there exists $x$ in $A_2$ such that for all $y_k$ in $B_1$, $xy_k \in E(G)$.  Then, $(A_1 \cup \{x\},B_1)$ is the maximum clique, contradicting the fact that $(A_1,B_1)$ is maximum.  Similar argument is true for $y \in B_2$.  Therefore, the lemma. \qed 
\end{proof}
\begin{lemma}
\label{lem2}
Let $G$ be a $P_5$-free chordal bipartite graph.  Then,  $\forall x \in A_{2}$, $N_G(x) \subset B_{1}$ and $ \forall y \in B_{2}$, $N_G(y) \subset A_{1}$. 
\end{lemma}
\begin{proof}
On the contrary,  $\exists x_{a} \in A_{2}$, $N(x_{a}) \not\subset B_{1}$. Case 1: $N(x_{a}){=}B_{1}$. Then, $(A_{1} \cup \{x_{a}\}, B_{1})$ is the  maximum biclique, a contradiction.  Case 2: $N(x_{a}) \subseteq B_{2}$.  This implies that there exists $y_{b} \in B_2$ such that $y_b \in N(x_{a})$.  Since $G$ is connected, $N(y_{b}) \cap A_{1} \neq \emptyset$, say $x_{c} \in N(y_{b}) \cap A_{1}$.  In $G$, $P(x_{a},x_{k})=(x_{a},y_{b},x_{c},y_{k},x_{k})$ is an induced $P_{5}$.  Note that, due to the maximality of $(A_1,B_1)$, as per Lemma \ref{lem1}, we find $x_{k} \not\in N(y_{b}), y_{k} \not\in N(x_{a})$.  This  contradicts that $G$ is $P_5$-free.  Case 3: $N(x_{a}) \cap B_{1} \neq \emptyset$ and $N(x_{a}) \cap B_{2} \neq \emptyset$. I.e., $ \exists x_{a} \in A_{2}$ such that $y_{b},y_{c} \in N(x_{a})$ and $y_{b} \in B_{1}$, $y_{c} \in B_{2}$.  In $G$, $P(y_{c},y_{k})=(y_{c},x_{a},y_{b},x_{k},y_{k}$), $x_{k} \not\in N(y_{c})$, $y_{k} \not\in N(x_{a})$ is an induced $P_{5}$.  Note that the existence of $x_k,y_k$ is due to Lemma \ref{lem1}.  This is contradicting the $P_5$-freeness of $G$.  Similarly, $ \forall y \in B_{2}$, $N(y) \subset A_{1}$ can be proved. \qed
\end{proof}
\begin{lemma}
\label{lem3}
Let $G$ be a $P_5$-free chordal bipartite graph.  
For $x_i,x_j \in A_{2}$, $i \not= j$, if $d_G(x_i) \leq d_G(x_j)$, then $N(x_i) \subseteq N(x_j)$.   Similarly, for $y_i,y_j \in B_{2}$, if $d_G(y_i) \leq d_G(y_j)$, $N(y_i) \subseteq N(y_j)$.
\end{lemma}
\begin{proof}
Let us assume to the contrary that $N(x_i) \not\subseteq N(x_j)$.  I.e., $N(x_i)  \setminus  N(x_j) \neq \emptyset$. \newline \textbf{Case 1:} $N(x_i) \cap N(x_j)\neq \emptyset$. I.e., $\exists y_{a}\in B$ such that  $y_{a} \in N(x_i) \cap N(x_j)$. Since $N(x_i) \not\subseteq N(x_j)$, $\exists y_{b}\in B$ such that  $y_{b} \not\in N(x_i) \cap N(x_j)$ and $y_{b} \in N(x_i)$. Since $d_G(x_j) \geq d_G(x_i)$, vertex $x_j$ is adjacent to at least one more vertex $y_{c} \in B$ such that $y_{c} \not\in N(x_i)$. The path $P(y_{c},y_{b})=(y_{c},x_j,y_{a},x_i,y_{b}$) is an induced $P_{5}$. This is a contradiction. \newline \textbf{Case 2:} $N(x_i) \cap N(x_j)= \emptyset $, $ | N(x_i) | \geq 1$ and $ | N(x_j) | \geq 1$. I.e., $\exists y_{a},y_{b}\in$B such that  $y_{a} \in N(x_i)$, $y_{a} \not\in N(x_j)$ and $y_{b} \in N(x_j)$, $y_{b} \not\in N(x_i)$.  Since $G$ is a connected graph,  $P^{'}(y_{a},y_{b}$)$\geq 3$, an induced path of length at least 3.  The path $P(x_i,x_j)=(x_i,P^{'}(y_{a},y_{b}),x_j$) has an  induced $P_{h}$,  $h \geq 5$. This is a contradiction.  Similarly, for all pairs of distinctive vertices $y_i,y_j \in B_{2}$ with $d_G(y_i) \leq d_G(y_j)$, $N(y_i) \subseteq N(y_j)$ can be proved.   \qed
\end{proof}
\begin{theorem}
\label{thm1}
Let $G$ be a $P_5$-free chordal bipartite graphs with $(A_1,B_1)$ being the maximum biclique.  Let $A_{2}=(u_{1},u_{2},...,u_{p})$ and $B_{2}=(v_{1},v_{2},...,v_{q})$ are orderings of vertices. If $d_{G}(u_1) \leq d_{G}(u_2) \leq d_{G}(u_3) \leq \ldots \leq d_{G}(u_p)$, then $N(u_{1}) \subseteq N(u_{2}) \subseteq N(u_{3}) \subseteq \ldots \subseteq N(u_{p})$.  Further, if  $d_{G}(v_1) \leq d_{G}(v_2) \leq d_{G}(v_3) \leq \ldots \leq d_{G}(v_q)$, then $N(v_{1}) \subseteq N(v_{2}) \subseteq N(v_{3}) \subseteq \ldots \subseteq N(v_{q})$.
\end{theorem} 
\begin{proof} 
We shall prove by mathematical induction on $|A_2|$. Base Case: $|A_{2}|=2, A_2=(u_{1},u_{2})$ such that $d_G(u_{1}) \leq d_G(u_2)$.  By Lemma \ref{lem3}, $N(u_1) \subseteq N(u_2)$.  Induction step: Consider $A_2=(u_{1},u_{2},u_{3},\ldots,u_{p}), p \geq 3$.  Consider the vertex $u_{p} \in A_{2}$ such that $d_G(u_p) \geq d_G(u_{p-1})$.  By Lemma \ref{lem3}, $N(u_{p-1}) \subseteq N(u_{p})$ is true.  By the hypothesis, $N(u_{1}) \subseteq N(u_{2}) \subseteq N(u_{3}) \subseteq \ldots \subseteq N(u_{p-1})$.   By combining the hypothesis and the fact that $N(u_{p-1}) \subseteq N(u_{p})$, our claim follows.  Similarly for $B_{2}$ as well.\qed
\end{proof}
We refer to the above ordering of vertices as {\em Nested Neighbourhood Ordering (NNO)} of $G$.  From now on, we shall arrange the vertices in $A_2$ in non-decreasing order of their degrees so that we can work with NNO of $G$.
\section{Hamiltonicity in $P_{5}$-free Chordal Bipartite graphs}
In this section, we shall present polynomial-time algorithms for hamiltonian cycle and path problems in $P_5$-free chordal bipartite graphs.   For a connected graph $G$ and set $S \subset V(G)$, $c(G-S)$ denotes the number of connected components in the graph induced on the set $V(G) \setminus S$.  It is well-known, due to, Chvatal \cite{dbwest2003} that if a graph $G$ has a hamiltonian cycle, then for every $S\subset V(G), c(G-S) \leq |S|$.  
\begin{theorem}
\label{thm2}
Let $G$ be a $P_5$-free chordal bipartite graph.  $G$ has a  hamiltonian cycle if and only if $|A| = |B|$ and $A_{2}$ has an ordering $(u_{1},u_{2},\ldots,u_{p})$, such that $\forall u_g, d_G(u_g) > g$, $1 \leq g \leq p$ and $B_{2}$ has an ordering $(v_{1},v_{2},\ldots,v_{q})$, $ \forall v_h, d_G(v_h) >h$, $1 \leq h \leq q$.
\end{theorem}
\begin{proof}
{\em Necessity:} On the contrary, $ \exists u_g \in A_{2}$  such that $u_g$ is the first vertex in the ordering with $d_G(u_g) \leq g$.  That is, for $u_k \in \{u_1,\ldots,u_{g-1}\}$, $d_{G}(u_k) > k$ and $d_G(u_g) \leq g$.   Let $d_{G}(u_g)= s$.  From Theorem \ref{thm1},  we know that $N(u_1) \subseteq N(u_2) \subseteq \ldots \subseteq N(u_{g-1}) \subseteq N(u_g)$.   This implies that $c(G-N(u_g)) = s+1 >s$.  This is a contradiction to Chvatal's necessary condition for hamiltonian cycle. Similarly, in $B_{2}$, $ \forall v_h,$ $d_G(v_h)>h$ can be proved. \\
{\em Sufficiency:} Let $i=|A_1|$ and $j=|B_1|$. Since $A_{2}$ has an ordering such that $ \forall u_g \in A_2$, $d_G(u_g)>g$, for clarity purpose, we define $N_G(u_g)$ as follows; $N(u_g)=\{y_{1},y_{2},\ldots,y_{l}\}$, $g<l<j$,  that is, $u_1$ is adjacent to at least two vertices $\{y_{1},y_{2}\}$ and at most $j-1$ vertices $\{y_{1},y_{2},\ldots,y_{j-1}\}$, $u_2$ is adjacent to at least three vertices $\{y_{1},y_{2},y_{3}\}$ and at most $j-1$ vertices $\{y_{1},y_{2},\ldots,y_{j-1}\}$ and similarly $u_p$ is adjacent to at least $p+1$ vertices $\{y_{1},y_{2},\ldots,y_{p+1}\}$ and at most $j-1$ vertices $\{y_{1},y_{2},\ldots,y_{j-1}\}$.  Observe that, due to the maximality of $(A_1,B_1)$, any $u_g$ of $A_2$ can be adjacent to at most $j-1$ vertices of $B_1$.  Similarly, in $B_2$, for all $v_h \in B_2$,               $N(v_h)=\{x_{1},x_{2},\ldots,x_{l}\},$ $h<l<i$. \\
Let $d_G(u_p)=r, p+1 \leq r \leq j-1$ and $d_G(v_q)=s, q+1 \leq s \leq i-1$. The vertices in $A_{1}$ can be ordered as $(x_{1},x_{2},\ldots,x_{q},x_{q+1},\ldots,x_{i})$ and the vertices in $B_{1}$ can be ordered as $(y_{1},y_{2},\ldots,y_{p},y_{p+1},\ldots,y_{j})$.  
Note that $A_{3}=A_{1}{\setminus}\{x_{1},x_{2},\ldots,x_{q+1}\}=\{x_{q+2},\ldots,x_{i-1},x_{i}\}$ and $B_{3}=B_{1}{\setminus}\{y_{1},y_{2},\ldots,y_{p+1}\}=\{y_{p+2},\ldots,y_{j-1},y_{j}\}$. 
Further, $|A_{3}| = | A | -(| A_{2} | +q+1)= | A | -(p+q+1)$ and $ | B_{3} | = | B | -( | B_{2} | +p+1)= | B | -(q+p+1)$. 
Since $ | A | = | B | $, it follows that $ | A_{3} | = | B_{3} | $. 
In $G$, \\ $(y_{1},u_{1},y_{2},u_{2},\ldots,y_{p},u_{p},y_{p+1},x_{1},v_{1},x_{2},v_{2},\ldots,x_{q},v_{q},x_{q+1},y_{p+2},x_{q+2},\ldots,y_{j},x_{i},y_{1})$ is a hamiltonian cycle. 
\end{proof}
The following lemma is well known and is due to Chvatal\cite{dbwest2003}. 
\begin{lemma}Chvatal\cite{dbwest2003}
If a graph $G$ has a Hamiltonian path, then for every $S\subset V(G), c(G-S) \leq | S|+1$.
\end{lemma}
\begin{theorem}
\label{thm3}
Let $G$ be a $P_5$-free chordal bipartite graph.  $G$ has a  hamiltonian path  if and only if one of the following is true \\
(i)  $|A|=|B|$ and $A_2$ has an ordering,  $\forall u_g, d_G(u_g) \geq g$, $1 \leq g \leq p$ and $B_{2}$ has an ordering, $ \forall v_h, d_G(v_h) \geq h$, $1 \leq h \leq q$. \\
(ii) $|A|=|B|+1$ and $A_2$ has an ordering,$\forall u_g, d_G(u_g) \geq g$, $1 \leq g \leq p$ and $B_{2}$ has an ordering, $ \forall v_h, d_G(v_h) >h$, $1 \leq h \leq q$. \\
\end{theorem}
\begin{proof}
{\em Necessity:} (i) Assume to the contrary that $ \exists u_g \in A_{2}$ such that $u_g$ is the first vertex in the ordering such that $d_G(u_g)<g$. Let $d_{G}(u_g)= s$. From Theorem \ref{thm1}, we know that $N(u_1) \subseteq N(u_2) \subseteq \ldots \subseteq N(u_{g-1}) \subseteq N(u_g)$.  Note that, as per the ordering of $A_2$, $N(u_{g-1})=N(u_g)$.  This implies that $c(G-N(u_g)) =  s+1+1 = s+2$.  Clearly, $s+2 \not \leq s+1$.  Thus, we contradict Chvatal's necessary condition for hamiltonian path.
Similarly $B_{2}$ has an ordering such that $ \forall v_h, d_G(v_h) \geq h$, $1 \leq h \leq q$. \\
(ii) For $A_2$, the argument is similar to the above.   Suppose ${\exists}v_{r}{\in}B_{2}$ such that $v_{r}$ is the first vertex in the ordering such that $d_G(v_{r}){\leq}r$. From Theorem \ref{thm1}, $N(v_{1}){\subseteq}N(v_{2}){\subseteq}\ldots{\subseteq}N(v_{r-1}){\subseteq}N(v_{r})$. Consider the set $S=B_{1}{\cup}\{v_{r+1},v_{r+2},\ldots,v_{q}\}$ and $|S|=j+q-r-1+1=j+q-r$.  Further, $c(G-S) {\geq} p+1+i-r$.  Note that $|A|=i+p$ and $|B|=j+q$.  Since $|A|=|B|+1$, $c(G-S) {\geq} p+1+i-r = j+q+1+1-r=j+q-r+2$.  Clearly, $c(G-S) \not \leq |S|+1$, contradicting Chvatal's condition for hamiltonian path. \\
{\em Sufficiency:} (i) Let $N(u_g)=\{y_{1},\ldots,y_{l}\}$, $g{\leq}l<j$ and $N(v_h)=\{x_{1},\ldots,x_{l}\},$ $h{\leq}l<i$. Consider $A_{3}=A_{1}{\setminus}\{x_{1},x_{2},\ldots,x_{q}\}=\{x_{q+1},x_{q+2},\ldots,x_{i-1},x_{i}\}$.  $B_{3}=B_{1}{\setminus}\{y_{1},y_{2},\ldots,y_{p}\}=\{y_{p+1},x_{p+2},\ldots,y_{j-1},y_{j}\}$.   Note that $|A_3|=|A|-(p+q)$ and $|B_3|=|B|-(p+q)$.  In $G$, \\$P(u_{1},y_{1},u_{2},y_{2},\ldots,u_{p},y_{p},x_{q+1},y_{p+1},x_{q+2},y_{p+2},\ldots,x_{i},y_{j},x_{q},v_{q},\ldots,x_{1},v_{1})$ is a hamiltonian path.  \\
(ii) Consider $A_{3}=A_{1}{\setminus}\{x_{1},x_{2},\ldots,x_{q+1}\}=\{x_{q+2},x_{q+3},\ldots,x_{i-1},x_{i}\}$ and 
    $B_{3}=B_{1}{\setminus}\{y_{1},y_{2},\ldots,y_{p}\}=\{y_{p+1},x_{p+2},\ldots,y_{j-1},y_{j}\}$.
In $G$, $P(u_{1},y_{1},u_{2},y_{2},\ldots,u_{p},y_{p},x_{q+2},y_{p+1},x_{q+3},y_{p+2},\ldots,x_{i},y_{j},x_{q+1},v_{q},\ldots,x_{2},v_{1},x_{1})$ is a hamiltonian path.  This completes a proof of this claim.
\end{proof}
\begin{theorem}
Let $G$ be a $P_5$-free chordal bipartite graph.  Finding hamiltonian path and cycle in $G$ are polynomial-time solvable.
\end{theorem}
\begin{proof}
Follows from  the characterizations presented in Theorems \ref{thm2} and \ref{thm3} as proofs are constructive.
\end{proof}
\section{Longest paths in $P_{5}$-free chordal bipartite graphs}
For a connected graph $G$, the longest path is an induced path of maximum length in $G$.  Since hamiltonian path is a path of maximum length, finding a longest path is trivially solvable if the input instance is an yes instance of hamiltonian path problem.  Thus, the longest path problem is a generalization of hamiltonian path problem, and hence the longest path problem is NP-complete if hamiltonian path problem is NP-complete in the graph class under study.  On the other hand, it is interesting to investigate the complexity of longest path problem in graphs where the hamiltonian path problem is polynomial-time solvable.   Since, hamiltonian path problem in $P_{5}$-free chordal bipartite graphs is polynomial-time solvable, in this section, we shall investigate the complexity of the longest path problem in $P_{5}$-free chordal bipartite graphs.  \\
{\bf Pruning:} We shall now prune $G$ by removing vertices that will not be part of any longest path in $G$.  Without loss of generality, we assume that $G$ has no hamiltonian path, and hence, there must exist vertices in $A_2$ ($B_2$) that violate degree conditions mentioned in Theorem \ref{thm3}.  As part of pruning, we prune such vertices from $G$.   Recall that $A_2=(u_{1},u_{2},\ldots,u_{p})$.  Let  $u_r$ be the first vertex in $A_2$ with $d_G(u_{r})<r$.  Remove $u_r$ and relabel the vertices of $A_2$ so that the sequence is reduced to $(u_{1},u_{2},\ldots,u_{p-1})$.  After, say $c$ iterations, $A_2$ becomes $(u_{1},u_{2},\ldots,u_{p-c})$ such that for $\forall u_g, 1 \leq g \leq (p-c), d_G(u_g) \geq g$.  Similarly, after $d$ iterations, $B_2$ becomes $(v_{1},v_{2},\ldots,v_{q-d})$ such that for $\forall v_h, 1 \leq h \leq (q-d), d_G(v_h) \geq h$.  From now on, when we refer to $A_2$ ($B_2$), it refers to the modified $A_2$ ($B_2$).  We define a subgraph $H_1$ on the modified $A_2$ and it is induced on the set $V(H_1)=A_2 \cup N(A_2)$. \\
Note that if $d_G(u_{p-c}) > p-c$, then we get a path $P_1=P(s_{a}=y_{1},u_{1},y_{2},u_{2},\ldots,y_{p-c},u_{p-c},y_{(p-c)+1}=t_{a})$ in $H_1$ and $|P_1| = 2(p-c)+1$.  If $d_G(u_{p-c}) = p-c$, then we get a path $P_1=P(s_{a}=u_{1},y_{1},u_{2},y_{2},\ldots,u_{p-c},y_{p-c}=t_{a})$ in $H_1$ and $|P_1| = 2(p-c)$.  \\ \\
{\bf Claim:} $P_1$ is a longest path in $H_1$. \\
Observe that if $d_G(u_{p-c}) > p-c$, then the subgraph of $H_1$ induced on the set $A_2 \cup \{y_{1},\ldots,y_{(p-c)+1} \}$ is an yes instance of the hamiltonian path problem.  Similarly, if $d_G(u_{p-c}) = p-c$, then the subgraph of $H_1$ induced on the set $A_2 \cup \{y_{1},\ldots,y_{(p-c)} \}$ is an yes instance of the hamiltonian path problem.  Since $H_1$ respects NNO,  for all the pruned vertices of $A_2$, their neighborhood is a subset of $\{y_{1},\ldots,y_{(p-c)+1} \}$.  This shows that the pruned vertices of $A_2$ cannot be augmented to $P_1$ to get a longer path in $H_1$.  This proves that $P_1$ is a longest path in $H_1$. \\ \\
We define a subgraph $H_2$ on the modified $B_2$ and it is induced on the set $V(H_2)=B_2 \cup N(B_2)$. \\
Note that if $d_G(v_{q-d}) > q-d$, then we get a path $P_2=P(s_{b}=x_{1},v_{1},x_{2},v_{2},\ldots,x_{q-d},v_{q-d},x_{(q-d)+1}=t_{b})$ in $H_2$ and $|P_2| = 2(q-d)+1$.  If $d_G(v_{q-d}) = q-d$, then we get a path $P_2=P(s_{b}=v_{1},x_{1},v_{2},x_{2},\ldots,v_{q-d},x_{q-d}=t_{b})$ in $H_2$ and $|P_2| = 2(q-d)$.  \\ \\
{\bf Claim:} $P_2$ is a longest path in $H_2$. \\
Observe that if $d_G(v_{q-d}) > q-d$, then the subgraph of $H_2$ induced on the set $B_2 \cup \{x_{1},\ldots,x_{(q-d)+1} \}$ is an yes instance of the hamiltonian path problem.  Similarly, if $d_G(v_{q-d}) = q-d$, then the subgraph of $H_1$ induced on the set $B_2 \cup \{x_{1},\ldots,x_{(q-d)} \}$ is an yes instance of the hamiltonian path problem.  Since $H_2$ respects NNO,  for all the pruned vertices of $B_2$, their neighborhood is a subset of $\{x_{1},\ldots,x_{(q-d)+1} \}$.  This shows that the pruned vertices of $B_2$ cannot be augmented to $P_2$ to get a longer path in $H_2$.  This proves that $P_2$ is a longest path in $H_2$. \\ \\
We define a subgraph $H_3$ induced on the set $V(H_3)=(~A_1 \setminus V(P_2)~) \cup (~B_1 \setminus V(P_1)~)$.  Since the subgraph induced on $(A_1,B_1)$ is a complete bipartite graph, $H_3$ is also a complete bipartite graph.  Let $A_1'=A_1 \setminus V(P_2)$ and $B_1'=(~B_1 \setminus V(P_1)~)$.  Let $P_3$ be a longest path in $H_3$.  We shall construct a longest path $P$ in $G$ using longest paths $P_1=(s_a,\ldots,t_a),P_2=(s_b,\ldots,t_b)$ and $P_3=(s_c,\ldots,t_c)$.  We use $(s_a,\ldots,t_a)$ and $(t_a,\ldots,s_a)$ interchangeably to refer to $P_1$, similarly, for $P_2$ and $P_3$ as well. \\
{\bf Case 1:} $|A_1'|=|B_1'|$.  Assume $s_c \in A_1'$ and $t_c \in B_1'$. \\
Case 1.1: $s_a,t_a \in B_1$.   Then, $P=(s_b,\ldots,t_b,t_c,\ldots,s_c,s_a,\ldots,t_a)$. \\
Case 1.2: $s_a \in A_2, t_a \in B_1$.  Then, $P=(s_a,\ldots,t_a,s_c,\ldots,t_c,t_b,\ldots,s_b)$. \\
{\bf Case 2:} $|A_1'|>|B_1'|$.  Note that $s_c,t_c \in A_1'$. \\
Case 2.1: $s_a,t_a \in B_1$. Then, $P=(s_b,\ldots,t_b,t_a,\ldots,s_a,s_c,\ldots,t_c)$. \\
Case 2.2: $s_a \in A_2, t_a \in B_1$.  Then, $P=(s_b,\ldots,t_b,t_{c-1},\ldots,s_c,t_a,\ldots,s_a)$, where $t_{c-1}$ is the vertex before $t_c$ in the ordering $(s_c,\ldots,t_c)$. \\
{\bf Case 3:} $|B_1'|>|A_1'|$.  Note that $s_c,t_c \in B_1'$ and $t_{c-1}$ is the vertex before $t_c$ in the ordering $(s_c,\ldots,t_c)$. \\
Case 3.1: $s_a,t_a \in B_1$.  Then, $P=(s_b,\ldots,t_b,s_c,\ldots,t_{c-1},s_a,\ldots,t_a)$. \\
Case 3.2: $s_a \in A_2, t_a \in B_1$.  Then, $P=(s_b,\ldots,t_b,s_c,\ldots,t_{c-1},t_a,\ldots,s_a)$. \\ \\
{\bf Claim:} $P$ is a longest path in $G$.  Further, $P$ can be computed in polynomial time.\\
Since $P_1, P_2$ and $P_3$ are longest paths as per above claims, $P$ is a longest path in $G$.  As our proofs are constructive in nature, we obtain $P$ in polynomial time.   \\ \\
{\bf Remark:} As an extension of longest path problem, we naturally obtain a minimum leaf spanning tree of $G$, which is a spanning tree of $G$ with the minimum number of leaves, in polynomial time.  Since $G$ respects NNO, the vertices pruned while constructing $H_1$ and $H_2$, cannot be included as internal vertices of $P$.  We shall now construct a minimum leaf spanning tree $T$ with $P$ as a subtree.  (i) the pruned vertices are augmented to $P$ as leaves to obtain $T$.  (ii) if $|A_1'|>|B_1'|$, then the vertices in $A_1'$ not included in $P_3$ are added to $P$ as leaves to obtain $T$.  Similarly, if $|B_1'|>|A_1'|$, then the vertices in $B_1'$ not included in $P_3$ are added to $P$ as leaves to obtain $T$.   This shows that $T$ has $|V(G) \setminus V(P)| + 2$ leaves. \\ \\
\section{Steiner path in $P_{5}$-free chordal bipartite graphs}
We now generalize the hamiltonian path problem and ask the follwing: given $G, R \subset V(G)$, find a path containing all of $R$ minimizing $V(G) \setminus R$, if exists.  Note that, this {\em constrained path problem} is the hamiltonian path problem when $R=V(G)$.  This has another motivation as well.  For a connected graph $G, R \subset V(G)$, the well-known Steiner tree problem asks for a tree containing all of $R$ minimizing $V(G) \setminus R$.  If we ask for a path instead of a tree in the Steiner tree problem, then it is precisely constrained path problem.  Due to this reason, we refer to this problem as the Steiner path problem.  It is important to highlight that not all input graphs have a solution to the Steiner path problem.   We shall present a constructive proof for the existence of Steiner path by case analysis. \\
{\bf Case 1:} $R \subseteq A_2$. \\
{\bf Claim:} If $R=(u_1,\ldots,u_r) \subseteq A_2$ is such that for all $u_g, 1 \leq g \leq r$, $d_G(g) \geq g$, then there exists a Steiner path in $G$.  Otherwise, no Steiner path exists in $G$. \\
Observe that, the vertices in $R$ respects degree constraints, due to which they respect NNO.  Thus, the graph induced on $R \cup N(R)$ is an yes instance of the hamiltonian path problem.  Therefore, $P=(u_1,y_1,\ldots,y_{r-1},u_r\}$ is a Steiner path.  Since $R$ is an independent set of size $r$, any path containing $R$ must have $r-1$ additional vertices and hence $P$ is a minimum Steiner path. \\ \\
{\bf Case 2:} $R \subseteq A_1$. \\
Let $R=(x_1,\ldots,x_r)$. It is easy to see that if $|R| \leq |B_1|+1$, then $P=(x_1,y_1,\ldots,x_{r-1},y_{r-1},x_r\}$ is a minimum Steiner path.  We shall work with $|R| > |B_1|+1$. \\
{\bf Claim:} If there exists $(z_1,\ldots,z_{r-(j+1)})$ in $B_2$ such that $(z_1,\ldots,z_{r-(j+1)})$ has NNO with $(w_1,\ldots,w_{r-(j+1)})$ of $A_1$, then there exists a Steiner path.  Otherwise, no Steiner path exists in $G$. \\
Note that the path $P=(z_1,w_1,\ldots,z_{r-(j+1)},w_{r-(j+1)},y_1,x_1,\ldots,y_j,x_{j+1})$ is a Steiner path of minimum cardinality. \\\\
{\bf Case 3:} $R \cap A_1 \not= \emptyset$ and $R \cap B_2 \not= \emptyset$. \\
{\bf Claim:} If $R \cap B_2=(z_1,\ldots,z_l)$ has NNO with $(w_1,\ldots,w_l)$ of $A_1$, then there exists a Steiner path in $G$.  Otherwise, no Steiner path exists in $G$. \\
The path $P=(z_1,w_1,\ldots,z_l,w_l,y_1,x_1,\ldots,y_{r-l-1},x_{r-l})$ is a minimum Steiner path in $G$. \\ \\
On the similar line, other cases $R \subseteq B_1$, $R \subseteq B_2$, $R \cap A_1 \not= \emptyset$ and $R \cap A_2 \not= \emptyset$, $R \cap A_2 \not= \emptyset$ and $R \cap B_1 \not= \emptyset$, $R \cap A_1 \not= \emptyset$ and $R \cap B_2 \not= \emptyset$ and  $R \cap A_2 \not= \emptyset$ can be proved. \\ \\
{\bf Conclusions and Further Research:} In this paper, we have presented structural results on $P_5$-free chordal bipartite graphs.  Subsequently, using these results, we have presented polynomial-time algorithms for hamiltonian cycle, hamiltonian path, longest path and Steiner path problems.  These results exploits the nested neighborhood ordering of $P_5$-free chordal bipartite graphs which is an important contribution of this paper.  A natural direction for further research is to study $P_6$-free chordal bipartite graphs and $P_7$-free chordal bipartite graphs as the complexity of these problems are open in these graph classes.

\end{document}